\documentclass[DIV12,11pt]{scrartcl}

\usepackage{amssymb,amsmath}
\usepackage{tikz,enumerate,cite,graphicx}
\usepackage{amsthm}
\usepackage{calc}
\usetikzlibrary{positioning}

\newcommand{\argmax}{\mathop{\operatorname{argmax}}}

\setkomafont{title}{\normalfont \LARGE}
\setkomafont{section}{\normalfont \Large \bfseries \boldmath}
\setkomafont{subsection}{\normalfont \large \bfseries\boldmath}
\setkomafont{subsubsection}{\normalfont \normalsize \bfseries}
\setkomafont{descriptionlabel}{\itshape}
\setkomafont{paragraph}{\normalfont \bfseries \boldmath}

\DeclareMathOperator{\probab}{\mathbb{P}}

\newcommand{\eps}{\varepsilon}
\newcommand{\cev}[1]{\reflectbox{\ensuremath{\vec{\reflectbox{\ensuremath{#1}}}}}}

\newcommand{\isoevent}[2]{\ensuremath{\underline{\mathsf E}_{#1}^{#2}}}
\newcommand{\isoeventp}[2]{\ensuremath{\overline{\mathsf E}_{#1}^{#2}}}

\newtheorem{theorem}{Theorem}[section]
\newtheorem{lemma}[theorem]{Lemma}
\newtheorem{remark}[theorem]{Remark}
\newtheorem{corollary}[theorem]{Corollary}
\newtheorem{claim}[theorem]{Claim}

\title{Smoothed Analysis of Belief Propagation for Minimum-Cost Flow and Matchingg\thanks{This research was supported by ERC Starting Grant 306465 (BeyondWorstCase)
and NWO grant 613.001.023 (Smoothed Analysis of Belief Propagation).}}

\author{Tobias Brunsch$^1$ \and Kamiel Cornelissen$^2$ \and Bodo Manthey$^2$ \and Heiko R\"oglin$^1$}
\date{\small \parbox{63mm}{\centering University of Bonn \\
Department of Computer Science \\ \texttt{brunsch@cs.uni-bonn.de, heiko@roeglin.org}}
\qquad \parbox{63mm}{\centering University of Twente \\ Department of Applied Mathematics \\
\texttt{b.manthey@utwente.nl, k.cornelissen@utwente.nl}}}

\begin{document}

\maketitle

\begin{abstract}
Belief propagation (BP) is a message-passing heuristic for statistical inference
in graphical models such as Bayesian networks and Markov random fields.
BP is used to compute marginal distributions or maximum likelihood assignments
and has applications in many areas, including machine learning, image
processing, and computer vision. However, the theoretical understanding of the
performance of BP is unsatisfactory.

Recently, BP has been applied to combinatorial optimization problems.
It has been proved that BP can be used to compute maximum-weight matchings and minimum-cost
flows for instances with a unique optimum. The number of iterations needed for this 
is pseudo-polynomial and hence BP is not efficient in general.   

We study belief propagation in the framework of smoothed analysis and prove that
with high probability the number of iterations needed to compute maximum-weight matchings
and minimum-cost flows is bounded by a polynomial if the weights/costs of the edges are
randomly perturbed. To prove our upper bounds, we use an isolation lemma by Beier and V\"{o}cking (SIAM J. Comput. 2006) for matching and generalize an isolation lemma for min-cost flow by Gamarnik, Shah, and Wei (Operations Research, 2012). We also prove almost matching lower tail bounds for the number of iterations that BP needs to converge.
\end{abstract}

\section{Belief Propagation}
\label{sec:bp}
The belief propagation (BP) algorithm is a message-passing algorithm that is used for solving probabilistic inference problems on graphical models. It has been introduced by Pearl in 1988~\cite{Pearl:PlausibleInference:1988}. Typical graphical models to which BP is applied are Bayesian networks and Markov random fields. There are two variants of
BP. The sum-product variant is used to compute marginal probabilities. The max-product or min-sum variant is used to compute maximum a posteriori (MAP) probability estimates.

Recently, BP has experienced great popularity. It has been applied in a large number of fields, such as machine learning, image processing, computer vision, and statistics. For an introduction to BP and several applications, we refer to Yedidia et al.~\cite{UnderstandingBP}. The are basically two main reasons for the popularity of BP. First of all, it is generally applicable and easy to implement because of its simple and iterative message-passing nature. In addition, it performs well in practice in numerous applications~\cite{Yanover02approximateinference,BPforStereo}.  

If the graphical model is tree-structured, BP computes exact
marginals/MAP estimates. In case the graphical model contains cycles,
convergence and correctness of BP have been shown only for specific classes of
graphical models. To improve the general understanding of BP and to gain new
insights about it, the performance of BP as either a heuristic or an exact algorithm
for several combinatorial optimization problems has been studied. Amongst others it has been
applied to the maximum-weight matching (MWM) problem, the minimum spanning tree (MST) problem,
the minimum-cost flow (MCF) problem, and the maximum-weight independent set problem~\cite{BPIndSet}.

Bayati et al.~\cite{BipartiteMatching} have shown that max-product BP
correctly computes the MWM in bipartite graphs in pseudo-polynomial time if it is
unique. For MST, it is known that BP converges to the correct optimal solution, if it converges at all (which is not guaranteed)~\cite{CavityMST}.
Gamarnik et al.~\cite{BeliefMCF} have shown that the max-product BP algorithm computes the
MCF in pseudo-polynomial time if it is unique.

\subsection{Belief Propagation for Matching and Flow Problems}

\label{sec:intromatch}
In this section we discuss the previous results about the BP algorithm for computing maximum-weight matchings and
minimum-cost flows in more detail. 
Bayati et al.~\cite{BipartiteMatching} have shown that the max-product BP algorithm correctly
computes the maximum-weight matching in bipartite graphs if it is unique. Convergence of BP takes pseudo-polynomial time and depends
linearly on the weight of the heaviest edge and on $1/ \delta$, where $\delta$
is the difference in weight between the best and second-best matching. 
In Section~\ref{sec:bpdef} we describe the BP algorithm for MWM in detail.

Belief propagation has also been applied to finding maximum-weight perfect
matchings in arbitrary graphs and to finding maximum-weight perfect
$b$-matchings~\cite{BayatiArbitraryMatching,SanghaviArbitraryMatching}, where a
perfect $b$-matching is a set of edges such that every vertex is incident to
exactly $b$ edges in the set. For arbitrary graphs the BP algorithm for MWM does not necessarily converge~\cite{SanghaviArbitraryMatching}. However, Bayati et al.~\cite{BayatiArbitraryMatching} and
Sanghavi et al.~\cite{SanghaviArbitraryMatching} have shown that
the BP algorithm converges to the optimal matching if the relaxation of the corresponding
linear program has an optimal solution that is unique and integer. The number of iterations needed
until convergence depends again linearly on the reciprocal of the parameter~$\delta$.
Bayati et al.~\cite{BayatiArbitraryMatching} have also shown that the same
result holds for
the problem of finding maximum-weight $b$-matchings that do not need to be perfect.

It turns out that BP can, to some extent, solve the relaxation
of the corresponding linear program for matching, even if it has a non-integral optimal solution. Bayati et
al.~\cite{BayatiArbitraryMatching} have shown that it is possible to solve
the LP relaxation by considering so-called \emph{graph covers}, in which they
compute a bipartite matching. In case of an optimum that is unique and integer, the optimal
solution in the graph cover corresponds to the optimal solution. In case of a
unique but fractional optimal solution, the average of the estimates of two
consecutive iterations (both of which are perfect matchings in the graph cover)
yield a value of $0$, $1/2$, or $1$ for any edge, which then equals its value in the optimal solution of the
relaxed LP. 

Sanghavi et al.~\cite{SanghaviArbitraryMatching} have shown that BP remains
uninformative for some edges (and outputs ``?'' for those), but computes the
correct values for all edges that have a fixed integral value in all optimal
solutions.

Gamarnik et al.~\cite{BeliefMCF} have shown that BP can be used to find a minimum-cost
flow, provided that the instance has a unique optimal solution. The number of iterations
until convergence is pseudo-polynomial and depends again linearly on the reciprocal of the difference in
cost between the best and second-best integer flow. In addition, they have proved a discrete isolation lemma~\cite[Theorem 8.1]{BeliefMCF} that shows that the edge costs can be slightly randomly perturbed to ensure that, with probability at least $1/2$, the perturbed MCF instance has a unique optimal solution. Using this result, they have constructed an FPRAS for MCF using BP.

\subsection{Smoothed Analysis}
Smoothed analysis has been introduced by Spielman and Teng~\cite{SpielmanTeng} in order to explain the performance of the simplex method for linear programming. It is a hybrid of worst-case and average-case analysis and an alternative to both: An adversary specifies an instance, and this instance is then slightly randomly perturbed. The perturbation can, for instance, model noise from measurement. Since its invention in 2001, smoothed analysis has been applied in a variety of contexts.
%~\cite{BV04,RT2009,ERV07,AMR09,LTRM,BR12}.
We refer to two recent surveys~\cite{CACM,it} for a broader picture.

We apply smoothed analysis to belief propagation for min-cost flow and maxi\-mum-weight matching. To do this, we consider the following general probabilistic model.
\begin{itemize}
\item The adversary specifies the graph $G=(V,E)$ and, in case of min-cost flow,
the integer capacities of the edges and the integer budgets 
(both are not required to be polynomially bounded).
Additionally the adversary specifies a probability density function $f_e:[0,1] \to [0,\phi]$ for every edge $e$.
\item The costs (for min-cost flow) or weights (for matching) of the edges are then drawn independently according to their respective density function.
\end{itemize}
The parameter $\phi$ controls the adversary's power: If $\phi=1$, then we have the average case. The larger $\phi$, the more powerful the adversary. The role of $\phi$ is the same as the role of $1/\sigma$ in the classical model of smoothed analysis, where instances are perturbed by independent Gaussian noise with standard deviation $\sigma$.
In that model the maximum density $\phi$ is proportional to $1/\sigma$.

\subsection{Our Results}
We prove upper and lower tail bounds for the number of iterations that BP needs to solve maximum-weight matching problems and min-cost flow problems. Our bounds match up to a small polynomial factor.

We prove that the probability that BP needs more than $t$ iterations is bounded
by $O(n^2 m\phi/t)$ for the min-cost flow problem and $O(nm \phi/t)$ for various matching problems, where $n$
and $m$ are the number of nodes and edges of the input graph, respectively
(Sections~\ref{sec:iso} and~\ref{sec:upper}). The upper bound for matching problems holds for the
variants of BP for the maximum-weight matching problem in bipartite graphs~\cite{BipartiteMatching}
as well as for the maximum-weight (perfect) $b$-matching problem in general
graphs~\cite{SanghaviArbitraryMatching,BayatiArbitraryMatching}. For the latter
it is required that the polytope corresponding to the relaxation of the matching LP is integral. If this is not the case,
we can still solve the relaxation of the matching LP with a slightly modified
BP algorithm~\cite{BayatiArbitraryMatching} (using graph covers, see Remark~\ref{rem:generalmatching}). To prove the upper tail bound for BP for MCF we use a continuous isolation lemma that is similar to the discrete isolation lemma by Gamarnik et al.~\cite[Theorem 8.1]{BeliefMCF}. We need the continuous version since we not only want to have that a unique optimal solution, but we also need to quantify the gap between the best and the second-best solution.
%\todo{Question~3:} \textbf{When Bayati et al. talk about solving the relaxation, a parameter $\eps(\lambda)$ appears in their paper. This seems to be
%different from the usual winner gap as it refers to the graph cover. Can we nevertheless apply the isolation lemma to bound $\eps(\lambda)$?
%If so we should elaborate more on that in the corresponding remark.}

These upper tail bounds are not strong enough to yield any bound on the expected number of iterations.
Indeed we show that this expectation is not finite by providing a lower tail bound of $\Omega(n \phi/t)$ for the 
probability that~$t$ iterations do not suffice
to find a maximum-weight matching in bipartite graphs.
This lower bound even holds in the average case, i.e., if $\phi = 1$ (Section~\ref{sec:uniformweights}),
and it carries over to the variants of BP for the min-cost flow problem and the minimum/maximum-weight (perfect) $b$-matching problem in general graphs mentioned
above~\cite{SanghaviArbitraryMatching,BayatiArbitraryMatching,BipartiteMatching,BeliefMCF}. The lower bound matches the upper bound up to a factor of $O(m)$ for
matching and up to a factor of $O(nm)$ for min-cost flow (Section~\ref{ssec:smoothedlower}). The smoothed lower bound even holds
for complete (i.e., non-adversarial) bipartite graphs.

Note that our lower bound on the number of iterations until convergence does not contradict the results by Salez and Shah~\cite{AverageMatchingBP}.
Roughly speaking, they have proved that BP for matching requires expected time $O(n^2)$. However, they allow that a small number of nodes are matched to incorrect nodes. It might even be the case that multiple nodes are matched to the same node. In our analysis we require convergence of the BP algorithm, i.e., each node should be matched to the unique node to which it is matched in the optimal matching. 

Finally, let us remark that, for the min-cost flow problem, we bound only the number of iterations that BP needs until
convergence. The messages might be super-polynomially long.
For all matching problems, however, the length of the messages is always bounded by a small polynomial.

\section{Definitions and Problem Statement}
In this section we define the maximum-weight matching problems that we consider and the min-cost flow problem. We also describe the BP algorithms that we apply. 

\subsection{Maximum-Weight Matching and Minimum-Cost Flow}
First we define the maximum-weight matching problem on bipartite graphs.
For this consider an undirected weighted bipartite graph $G=(U\cup V,E)$ with $U=\{u_1,\ldots,u_n\}$, $V=\{v_1,\ldots,v_n\}$, and $E\subseteq \{(u_i,v_j)=e_{ij}, 1 \leq i,j \leq n\}$. Each edge $e_{ij}$ has weight $w_{ij}\in \mathbb{R}^+$. A collection of edges $M \subseteq E$ is called a matching if each node of $G$ is incident to at most one edge in~$M$. We define the weight of a matching $M$ by
\[ \displaystyle w(M)=\sum_{e_{ij} \in M} w_{ij}. \]
The maximum-weight matching $M^\star$ of~$G$ is defined as
\[ 
M^\star=\argmax\{w(M)\mid M \textrm{ is a matching of }G\}. \]

A $b$-matching~$M\subseteq E$ in an arbitrary graph~$G=(V,E)$ is a set of edges such that every node from~$V$ is incident to at
most~$b$ edges from~$M$. A $b$-matching is called perfect if every node from~$V$ is incident to exactly~$b$ edges from~$M$.
Also for these problems we assume that each edge~$e\in E$ has a certain weight~$w_e$ and we define the weight of a $b$-matching~$M$
accordingly.

\subsection{Min-Cost Flow Problem}
In the min-cost flow problem (MCF) the goal is to find a cheapest flow that satisfies all capacity and budget constraints. 
We are given a graph $G=(V,E)$ with $V=\{v_1,\ldots,v_n\}$. In principle we allow multiple edges between a pair of nodes, but for ease of notation we consider simple 
directed graphs. Each node $v$ has a budget $b_v \in \mathbb{Z}$. Each directed edge $e=e_{ij}$ from $v_i$ to $v_j$ has capacity 
$u_e \in \mathbb{N}_0$ and 
cost $c_e \in \mathbb{R}^+$. For each node $v\in V$, we define $E_v$ as the set of edges incident to $v$. For each edge $e\in E_v$ we define $\Delta(v,e)=1$ if $e$ is an
out-going edge of $v$ and $\Delta(v,e)=-1$ if $e$ is an in-going edge of $v$. In the MCF one needs to assign a flow $f_e$ to each edge $e$ such that the total
cost $\sum_{e\in E} c_e f_e$ is minimized and the flow constraints
\[
0\leq f_e \leq u_e \quad \text{for all $e \in E,$}
\]
and budget constraints
\[
\sum_{e\in E_v} \Delta(v,e)f_e=b_v \quad \text{for all $v \in V$}
\]
are satisfied. We refer to Ahuja et al.~\cite{AhujaNetworkFlow} for more details about MCF.

Let us remark that we could have allowed also rational values for the budgets and capacities. 
As our results do not depend on these values, they are not affected by scaling all capacities and budgets by the smallest common denominator.

Note that finding a perfect minimum-weight matching in a bipartite graph $G=(U\cup V,E)$ is a special case of the
min-cost flow problem, see Ahuja et al.~\cite{AhujaNetworkFlow} for details.

\subsection{Belief Propagation}
\label{sec:bpdef}

For convenience, we describe the BP algorithm used by Bayati et al.~\cite{BipartiteMatching}. For the details of the other versions of BP for the (perfect) maximum-weight $b$-matching problem and the min-cost flow problem we refer to the original works~\cite{BayatiArbitraryMatching,SanghaviArbitraryMatching,BeliefMCF}. When necessary, we discuss the differences between the different versions of BP in Sections~\ref{sec:upper} and \ref{sec:lower}.

The BP algorithm used by Bayati et al.~\cite{BipartiteMatching} is an iterative
message-passing algorithm for computing maximum-weight matchings (MWM). Bayati et al.\ define their algorithm
for complete bipartite graphs $G=(U\cup
V,E)$ with $|U|=|V|=n$. In each iteration $t$, each node $u_i$ sends a
message vector
\[ \vec{M}^t_{ij}=[\vec{m}^t_{ij}(1),\vec{m}^t_{ij}(2),\ldots,\vec{m}^t_{ij}(n)] \]
to each of its neighbors $v_j$. The messages can be interpreted as how `likely' the sending node thinks it is that the receiving node should be matched to a particular node in the MWM. The greater the value of the message $\vec{m}^t_{ij}(r)$, the more likely it is according to node $u_i$ in iteration $t$ that node $v_j$ should be matched to node $u_r$. Similarly, each node $v_j$ sends a message vector $\cev{M}^t_{ji}$ to each of its neighbors $u_i$. The messages are initialized as
\begin{align*}
\vec{m}^0_{ij}(r) &= \begin{cases} w_{ij} & \text{if $r=i$,} \\ 0 & \text{otherwise and} \end{cases} \\
\cev{m}^0_{ji}(r) &= \begin{cases} w_{ij} & \text{if $r=j$,} \\ 0 & \text{otherwise.} \end{cases}
\end{align*}

The messages in iterations $t \geq 1$ are computed from the messages in the previous iteration as follows
\begin{align*}
\vec{m}^t_{ij}(r) &= \begin{cases} w_{ij} + \displaystyle\sum_{k \neq j} \cev{m}^{t-1}_{ki}(j) & \text{if $r=i$,} \\
\displaystyle\max_{q \neq j} \left[w_{iq}+\displaystyle\sum_{k \neq j} \cev{m}^{t-1}_{ki} (q)\right] & \text{otherwise and} \end{cases}  \\
\cev{m}^t_{ji}(r) &= \begin{cases} w_{ij} + \displaystyle\sum_{k \neq i} \vec{m}^{t-1}_{kj}(i) & \textrm{if $r=j$,} \\
\displaystyle\max_{q \neq i} \left[w_{qj}+\displaystyle\sum_{k \neq i} \vec{m}^{t-1}_{kj} (q)\right] & \text{otherwise.} \end{cases}
\end{align*}

The beliefs of nodes $u_i$ and $v_j$ in iteration $t$ are defined as 
\begin{align*}
b^t_{u_i}(r)&=w_{ir}+\sum_k \cev{m}^t_{ki}(r), \\
b^t_{v_j}(r)&=w_{rj}+\sum_k \vec{m}^t_{kj}(r).
\end{align*}
The beliefs can be interpreted as the `likelihood' that a node should be matched to a particular neighbor. The greater the value of $b^t_{u_i}(j)$, the more likely it is that node $u_i$ should be matched to node $v_j$. We denote the estimated MWM in iteration $t$ by $\tilde{M}^t$. The estimated matching $\tilde{M}^t$ matches each node $u_i$ to node $v_j$, where $j=\argmax_{ 1\leq r \leq n}\{b^t_{u_i}(r)\}$. Note that $\tilde{M}^t$ does not always define a matching, since multiple nodes may be matched to the same node. However, Bayati et al.~\cite{BipartiteMatching} have shown that if the MWM is unique, then for $t$ large enough, $\tilde{M}^t$ is a matching and equal to the MWM.   

\section{Isolation Lemma for Maximum-Weight Matchings and Min-Cost Flows}
\label{sec:iso}

Before we turn to proving the upper tail bounds for the number of iterations of the BP algorithm in Section~\ref{sec:upper}, we take a closer look at the
quantity~$\delta$, which we defined above as the difference in weight or cost between the best and second-best matching or
integer flow, respectively. The previous results discussed in Section~\ref{sec:intromatch} indicate that in order for the BP algorithm
to be efficient~$\delta$ must not be too small. While~$\delta$ can be arbitrarily small for weights or costs that are chosen by
an adversary, it is a well-known phenomenon that~$\delta$ is with high probability not too small when the weights or costs are drawn
randomly.

\subsection{Maximum-Weight Matchings}
\label{subsec:IsolationBinary}

Beier and V\"ocking~\cite{WinnersLosers} have considered a general scenario in which an arbitrary set~$S\subseteq\{0,1\}^m$ of feasible solutions
is given and to every $x=(x_1,\ldots,x_m)\in S$ a weight~$w\cdot x=w_1x_1+\ldots+w_mx_m$ is assigned by a linear objective
function. As in our model they assume that every coefficient~$w_i$ is drawn independently according to an adversarial density
function~$f_i:[0,1]\to[0,\phi]$ and they define~$\delta$ as the difference in weight between the best and the second-best feasible solution from~$S$,
i.e., $\delta=w\cdot x^\star - w\cdot \hat{x}$ where~$x^\star=\argmax_{x\in S}w\cdot x$ and~$\hat{x}=\argmax_{x\in S\setminus\{x^\star\}}w\cdot x$.
They prove a strong isolation lemma that, regardless of the adversarial choices of~$S$ and the density functions~$f_i$, the probability
of the event~$\delta\le\eps$ is bounded from above by~$2\eps\phi m$ for any~$\eps\ge 0$.

If we choose~$S$ as the set of incidence vectors of all matchings or (perfect) $b$-matchings in a given graph, Beier and V\"ocking's
results yield for every~$\eps\ge 0$ an upper bound on the probability that the difference in weight $\delta$ between the best and second-best matching 
or the best and second-best (perfect) $b$-matching is at most~$\eps$. Combined with the results in Section~\ref{sec:intromatch},
this can immediately be used to obtain an upper tail bound on the number of iterations of the BP algorithm for these problems.

\subsection{Min-Cost Flows}

The situation for the min-cost flow problem is significantly more difficult because the set~$S$ of feasible integer flows cannot naturally
be expressed with binary variables. If one introduce a variables for each edge corresponding to the flow on that edge, 
then~$S\subseteq\{0,1,2,\ldots,u_{\max}\}^m$ where~$u_{\max}=\max_{e\in E}u_e$. R\"oglin and V\"ocking~\cite{IntegerProgramming}
have extended the isolation lemma to the setting of integer, instead of binary, vectors. However, their result is not strong enough
for our purposes as it bounds the probability of the event~$\delta\le\eps$ by~$\eps\phi m (u_{\max}+1)^2$ from above for
any~$\eps\ge 0$. As this bound depends on~$u_{\max}$ it would only lead to a pseudo-polynomial upper tail bound on the number
of iterations of the BP algorithm when combined with the results of~\cite{BeliefMCF}. Our goal is, however, to obtain a polynomial
tail bound that does not depend on the capacities. In the remainder of this section, we prove that the isolation lemma from~\cite{IntegerProgramming}
can be significantly strengthened when structural properties of the min-cost flow problem are exploited.

In the following we consider the
residual network for a flow $f$. For each edge $e_{ij}$ in the original network
that has less flow than its capacity $u_{ij}$, we include an edge $e_{ij}$ with
capacity $u_{ij}-f_{ij}$ in the residual network. Similarly, for each edge
$e_{ij}$ that has flow greater than zero, we include the backwards edge $e_{ji}$
with capacity $f_{ij}$ in the residual network. We refer to Ahuja et
al.~\cite{AhujaNetworkFlow} for a more details about residual
networks.

As all capacities and budgets are integers, there is always a min-cost flow that is integral. An additional property of our probabilistic
model is that with probability one there do not exist two different integer flows with exactly the same costs. This follows directly from
the fact that all costs are continuous random variables. Hence, without loss of generality we restrict our presentation in the following
to the situation that the min-cost flow is unique.  

In fact, Gamarnik et al.~\cite{BeliefMCF} have not used~$\delta$, the difference in cost between the best and second-best integer
flow, to bound the number of iterations needed for BP to find the unique optimal solution of MCF, but they have used another quantity~$\Delta$. They have defined~$\Delta$ as the length of the cheapest cycle in the residual network of the min-cost flow~$f^\star$.
Note that~$\Delta$ is always non-negative. Otherwise, we could send one unit of flow along a cheapest cycle. This would result
in a feasible integral flow with lower cost. With the same argument we can argue that~$\Delta$ must be at least as large as~$\delta$
because sending one unit of flow along a cheapest cycle results in a feasible integral flow different from~$f^\star$ whose costs
exceed the costs of~$f^\star$ by exactly~$\Delta$. Hence any lower bound for~$\delta$ is also a lower bound for~$\Delta$ and so
it suffices for our purposes to bound the probability of the event~$\delta\le\eps$ from above. 
%\begin{lemma}
%\label{lem:deltaint}
%Let $f^\star$ be an optimal solution of a min-cost flow problem with integer capacities and budgets, and let $\hat f$ be the best integer flow with $\hat f \neq f^\star$. Then $\delta \geq c \cdot \hat f - c \cdot f^\star$.
%\end{lemma}

%\begin{proof}
%If the optimal flow $f^\star$ is not unique, then the residual network corresponding to $f^\star$ contains a cycle of length $0$ \todo{Reference?}, which implies $\delta = 0$. \todo{Why do we need the previous argument?} \todo{Don't we need $\delta \geq 0$ (which always holds) and $c \cdot \hat f - c \cdot f^\star = 0$?} Otherwise, $f^\star$ is unique and integral.

%Let $c$ be a cycle of cost $\delta$ in the residual network for the optimal flow $f^\star$. Because we have integer capacities and budgets, we can ship at least one unit of flow along $c$, which gives us another feasible integer flow $f'$. By construction, we have $c \cdot f' = c \cdot f^\star + \delta$. By the definition of $\hat f$, we have $c \cdot f' \geq c \cdot \hat f$, which proves the lemma.
%\end{proof}

The isolation lemma we prove is based on ideas that Gamarnik et al.~\cite[Theorem 8.1]{BeliefMCF} have 
developed to prove that the
optimal solution of a min-cost flow problem is unique with high probability if
the costs are randomly drawn integers from a sufficiently large set. We provide a
continuous counterpart of this lemma, where we bound the probability that the
second-best integer flow is close in cost to the optimal integer flow.

\begin{lemma}
\label{lem:isoint}
The probability that the cost of the optimal and the second-best integer flow differs by at most $\eps\ge 0$ is bounded 
from above by $2 \eps \phi m$.
\end{lemma}

\begin{proof}
Consider any fixed edge $\tilde e$, and let the costs of all other edges be fixed by an adversary. The cost $c_{\tilde e}$ of $\tilde e$ is drawn according to its probability distribution, whose density is bounded by $\phi$.

For~$e\in E$, let $\isoevent \eps e$ be the event that there exist two different integer flows $f^\star$ and $\hat f$ with the following properties:
\begin{enumerate}[(i)]
\item $f^\star $ is optimal.
\item $c \cdot f^\star$ and $c\cdot \hat f$ differ by at most $\eps$, i.e., $c \cdot \hat f \leq c \cdot f^\star + \eps$.
\item $f^\star_{e} = 0$ and $\hat f_{e} > 0$. \label{con:3}
\end{enumerate}
Let $\isoeventp \eps e$ be analogously defined, except for Condition~\eqref{con:3} being replaced by $f^\star_{e} = u_e$ and $\hat f_{e} < u_e$.

\begin{claim}
\label{cla:e}
Let $e \in E$ be arbitrary. Assume that all costs except for $c_{e}$ are fixed. Let $I \subseteq [0,1]$ be the set of real numbers such that $I = \{c_{e} \mid \isoevent \eps{e}\}$. Then $I$ is a subset of an interval of length at most $\eps$.
\end{claim}

\begin{proof}
%According to Gamarnik et al.~\cite[Theorem 8.1]{BeliefMCF}, there is at most one
%value $\alpha \in \mathbb{R}$ such that $c_{e} = \alpha$ leads to two
%optimal solutions $f^\star$ and $\hat f$ with $f^\star_{e} = 0$ and $\hat
%f_{e} > 0$.
If~$I\neq\emptyset$, let~$\alpha=\min(I)$ and let~$f^\star$ be an optimal integer
flow for~$c_{e}=\alpha$ with~$f^\star_{e}=0$.
Due to the choice of~$\alpha$ it is clear that~$I \subseteq [\alpha,\infty)$
We claim that $I \subseteq [\alpha,\alpha+\eps]$.  
If $c_{e}  = \alpha + \eta$ for some $\eta>0$, then
$f^\star$ stays optimal, and, for any feasible integer solution $f$ with
$f_{e} > 0$, we have
\begin{align*}
c \cdot f & = \sum_{\tilde e\neq e} c_{\tilde e} f_{\tilde e} + (\alpha + \eta)  f_{e} 
 \geq \sum_{\tilde e\neq e} c_{\tilde e} f_{\tilde e} + \alpha f_{e} + \eta & \text{as $f_{e} \geq 1$}\\
& \geq c \cdot f^\star + \eta & \text{as $f^\star_{e} = 0$ and $f^\star$ is optimal.}
\end{align*}
Thus, for $\eta > \eps$, the event $\isoevent \eps{e}$ does not occur.
%If, on the other hand, $c_{\tilde e} = \alpha - \eta < \alpha$ for some $\eta$, then we have, for any feasible integer flow $f$ with $f_{\tilde e}  =0$ (thus, in particular for $f^\star$),
%\begin{align*}
%c \cdot f & = \sum_{e \neq \tilde e} c_{e} f_e + (\alpha - \eta) f_{\tilde e} 
% = \sum_{e \neq \tilde e} c_{e} f_e + \alpha  f_{\tilde e} & \text{as $f_{\tilde e} = 0$}\\
%  & \geq \sum_{e \neq \tilde e} c_{e} \hat f_e + \alpha \hat f_{\tilde e} & \text{as $\hat f$ is optimal if $c_{\tilde e} = \alpha$} \\
%  & \geq c \cdot \hat f + \eta & \text{as $\hat f_{\tilde e} \geq 1$.}
%\end{align*}
%Thus, no $f$ with $f_{\tilde e} = 0$ is optimal in this case, and consequently event $\isoevent \eps{\tilde e}$ does not occur.
\end{proof}

The proof of the following claim is omitted as it is completely analogous to the proof of the previous claim.
\begin{claim}
\label{cla:ee}
Let $e \in E$ be arbitrary. Assume that all costs except for $c_{e}$ are fixed. Let $I \subseteq [0,1]$ be the set of real numbers such that $I = \{c_{e} \mid \isoeventp \eps{e}\}$. Then $I$ is a subset of an interval of length at most $\eps$.
\end{claim}

The following claim shows that, provided no event $\isoevent \eps
e$ or $\isoeventp \eps e$ occurs, the second-best integer flow is more
expensive than the best integer flow by at least an amount of $\eps$.

\begin{claim}
Assume that for every edge~$e\in E$ neither $\isoevent \eps e$ nor $\isoeventp \eps e$ occurs. Let $f^\star$ be a min-cost flow and let $\hat f \neq f^\star$ be a min-cost integer flow that differs from $f^\star$, i.e., a second-best integer flow.
Then $c \cdot \hat f \geq c\cdot f^\star + \eps$.
\end{claim}

\begin{proof}
First we prove that under our assumption that the min-cost flow is unique some edge $e$ exists such that
$f^\star_{e} \in \{0,u_{e}\}$ and $\hat f_{e} \neq f^\star_{e}$.
Suppose that no such edge $e$ exists and let $d = f^\star - \hat f$. Then $d_{e} > 0$ only if
$f^\star_{e} < u_{e}$ and $d_e < 0$ only if $f^\star_e > 0$ because otherwise there is
an edge $e$ with $f^\star_e \in \{0, u_e\}$ and $\hat f_e \neq f^\star_e$.
From this, we can conclude that there exists a $\lambda > 0$ such that $f^\star
+ \lambda d$ is a feasible flow. Let $\lambda_0 = \max\{\lambda \mid f^\star +
\lambda d \text{ is feasible}\}$ and $\check f = f^\star + \lambda_0 d$. 
From the assumption that the min-cost flow is unique it follows that~$c\cdot d = c\cdot f^\star - c\cdot \hat f < 0$.
Hence, $c\cdot \check f < c\cdot f^\star$, contradicting the choice of~$f^\star$ as min-cost flow.

This argument shows that there always exists an edge~$e$ such that $f^\star_{e} \in \{0,u_{e}\}$ and $\hat f_{e} \neq f^\star_{e}$.
As none of the events $\isoevent \eps e$ and $\isoeventp \eps e$ occurs for this edge~$e$, it follows that
$c \cdot \hat f \geq c\cdot f^\star + \eps$.
\end{proof}

From Claims~\ref{cla:e} and~\ref{cla:ee}, we obtain $\probab(\isoevent \eps e) \leq \eps \phi$ and $\probab(\isoeventp \eps e) \leq \eps \phi$: We fix all edge costs except for $c_e$ and then $\isoevent \eps e$ can only occur if $c_e$ falls into an interval of length at most $\eps$. Since the density function of $c_e$ is bounded from above by $\phi$, this happens with a probability of at most $\eps \phi$. The same holds for any $\isoeventp \eps e$. Thus, the lemma follows by a union bound over all $2m$ events $\isoevent \eps e$ and $\isoeventp \eps e$.
\end{proof}

The isolation lemma (Lemma~\ref{lem:isoint}) together with the discussion about the relation between~$\delta$,
the difference in cost between the best and second-best integer flow,
and~$\Delta$, the length of the cheapest cycle in the residual network of the min-cost flow~$f^\star$, 
immediately imply the following upper bound for the probability that $\Delta$ is small.

\begin{corollary}
\label{cor:deltabound}
For any $\eps > 0$, we have
$\probab(\Delta \leq \eps) \leq 2\eps \phi m$.
\end{corollary}

\section{Upper Tail Bounds}
\label{sec:upper}

\subsection{Maximum-Weight Matching}
\label{ssec:uppermatching}

We first consider the BP algorithm of Bayati et al.~\cite{BipartiteMatching}, which computes maximum-weight matchings
in complete bipartite graphs~$G$ in $O(nw^\star/\delta)$ iterations on all instances with a unique optimum. Here~$w^\star$ denotes the
weight of the heaviest edge and~$\delta$ denotes the difference in weight between the best and the second-best matching.
Even though it is assumed that~$G$ is a complete bipartite graph, this is not strictly necessary. If a non-complete graph
is given, missing edges can just be interpreted as edges of weight~$0$. 

With Beier and V\"ocking's isolation lemma (cf.\ Section~\ref{subsec:IsolationBinary}) we obtain the following tail bound for the number of iterations
needed until convergence when computing maximum-weight perfect matchings in bipartite graphs using
BP.

\begin{theorem}\label{thm:MatchingUpper}
Let $\tau$ be the number of iterations until Bayati et al.'s BP~\cite{BipartiteMatching} for maximum-weight perfect bipartite matching converges. Then $\probab(\tau \geq t) = O(nm\phi/t)$.
\end{theorem}
\begin{proof}
The number of iterations until BP converges is bounded from above by $O(nw^\star/\delta)$~\cite{BipartiteMatching}.
The weight of each edge is at most 1, so $w^\star \leq 1$. The upper bound exceeds $t$ only if $\delta \leq O(n/t)$.
By Beier and V\"ocking's isolation lemma, we have $\probab(\delta \leq O(n/t)) \leq O(nm\phi /t)$, which yields the bound claimed.
\end{proof}

This tail bound is not strong enough to yield any bound on the expected running-time of BP for bipartite matchings. But it is strong enough to show that BP has smoothed polynomial running-time with respect to the relaxed definition adapted from average-case complexity~\cite{WinnersLosers}, where it is required that the expectation of the running-time to some power $\alpha> 0$ is at most linear. However, a bound on the expected number of iterations is impossible, and the tail bound proved above is tight up to a factor of $O(m)$ (Section~\ref{sec:lower}).

As discussed in Section~\ref{sec:intromatch}, BP has also been applied to
finding maximum-weight (perfect) $b$-matchings in arbitrary
graphs~\cite{BayatiArbitraryMatching,SanghaviArbitraryMatching}. The result is
basically that BP converges to the optimal matching if the optimal solution of the relaxation of the
corresponding linear program is unique and integer. The number of
iterations needed until convergence depends again on ``how unique'' the optimal
solution is. For Bayati et al.'s variant~\cite{BayatiArbitraryMatching}, the
number of iterations until convergence depends on $1/\delta$, where $\delta$ is
again the difference in weight between the best and the second-best matching.
For Sanghavi et al.'s variant~\cite{SanghaviArbitraryMatching}, the number of
iterations until convergence depends on $1/c$, where $c$ is the smallest rate by
which the objective value will decrease if we move away from the optimum
solution.

However, the technical problem in transferring the upper bound for bipartite graphs to arbitrary graphs is that the adversary can achieve that, with high probability or even with a probability of $1$ (for larger $\phi$), the optimal solution of the LP relaxation is not integral. Already in the average-case, i.e., for $\phi =1$, where the adversary has no power at all, the optimal solution of the LP relaxation has some fractional variables with high probability.

Still, we can transfer the results for bipartite matching to both algorithms for arbitrary matching if we restrict the input graphs to be bipartite, since in this case the constraint matrix of the associated LP is totally unimodular.

\begin{theorem}
Let $\tau$ be the number of iterations until Bayati et al.'s~\cite{BayatiArbitraryMatching} or Sanghavi et al.'s~\cite{SanghaviArbitraryMatching} BP for general matching, restricted to bipartite graphs as input, converges. Then $\probab(\tau \geq t) = O(nm\phi/t)$.
\end{theorem}

\begin{proof}
For Bayati et al.'s BP algorithm, this follows in the same way as Theorem~\ref{thm:MatchingUpper} from their bound on the number of iterations until convergence, 
which is $O(n/\delta)$~\cite[Theorem 1]{BayatiArbitraryMatching}.

Sanghavi et al.\ prove that their variant of BP for general graphs converges after $O(1/c)$ iterations, provided that the LP relaxation has no fractional optimal solutions. Here, $c$ is defined as
\[
 c = \min_{\text{$\hat x \neq x^\star$ is a vertex of $P$}} \frac{w \cdot (x^\star - \hat x)}{\|x^\star - \hat x\|_1},
\]
where $x^\star$ is the (unique) optimal solution to the relaxation and $P$ is the matching polytope~\cite[Remark 2]{SanghaviArbitraryMatching}.

For any $\hat x \neq x^\star$, we have $\|x^\star - \hat x\|_1 \leq n$.
Furthermore, $w \cdot (x^\star - \hat x)$ is just the difference in weights
between $x^\star$ and $\hat x$. Since the input graph is bipartite, all vertices of $P$ are integral. Thus, $w \cdot (x^\star - \hat x) \geq \delta$,
where (again) $\delta$ is the difference in weight between the best and the
second-best matching.
Thus, $c \geq \delta/n$, which proves the theorem.
\end{proof}

\begin{remark}
\label{rem:generalmatching}
Bayati et al.~\cite{BayatiArbitraryMatching} and Sanghavi et al.~\cite{SanghaviArbitraryMatching} have also shown how to compute $b$-matchings with BP. If~$b$ is even, then the unique optimum to the LP relaxation is integral. Thus, we circumvent the problem that the optimal solution might be fractional. Hence, following the same reasoning as above, the probability that BP for $b$-matching for even~$b$ runs for more than~$t$ iterations until convergence is also bounded by $O(mn \phi/t)$.

Furthermore, Bayati et al.~\cite[Section 4]{BayatiArbitraryMatching} have shown how to compute the optimal solution of the relaxation of the matching LP with \emph{graph covers}. They obtain the same $O(n/\delta)$ bound for the number of iterations until convergence as for ordinary matching. However, since we are no longer talking about integer solutions, we cannot directly apply the
isolation lemma of Beier and V\"ocking~\cite{WinnersLosers}.
To see that $\delta$ is still unlikely to be small in the same way (with a slightly worse constant), we
can apply the isolation lemma of R\"oglin and V\"ocking~\cite{IntegerProgramming} since the matching polytope is half-integral.
Thus, if we scale the right-hand side with a factor of $2$, then we obtain a $0/1/2$ integer program.
Because of this, we obtain the same $O(mn \phi/t)$ tail bound for the probability that the number of iterations until convergence exceeds $t$.
\end{remark}
\subsection{Min-Cost Flow}
\label{ssec:upperflow}

The bound for the probability that $\Delta$ is small (Corollary~\ref{cor:deltabound}) plus the pseudo-polynomial
bound by Gamarnik et al.~\cite{BeliefMCF} yields a tail bound for the number of iterations that BP needs until convergence.

\begin{theorem}
\label{thm:mcfupper}
Let $\tau$ be the number of iterations until BP for min-cost flow~\cite{BeliefMCF} converges. Then $\probab(\tau \geq t) = O(n^2m\phi/t)$.
\end{theorem}

\begin{proof}
The number of iterations until BP for min-cost flow converges is bounded from above by $cLn/\Delta$ for some constant $c$, where $L$ is the maximum cost of a simple directed path in the residual network for the optimal flow~\cite[Theorem~4.1]{BeliefMCF}. The cost of each edge is at most 1, so $L\leq n$. The upper bound exceeds $t$ only if $\Delta \leq cn^2/t$. By Corollary~\ref{cor:deltabound}, we have $\probab(\Delta \leq cn^2/t) \leq 2c n^2m\phi /t$, which yields the bound claimed.
\end{proof}

\section{Lower Tail Bounds}
\label{sec:lower}

We show that the expected number of iterations necessary for convergence of BP
for maximum-weight matching (MWM) is unbounded. To do this, we prove a lower tail bound on the number of iterations that matches the upper tail bound from Section~\ref{sec:upper}. The lower bound holds even for a two by two complete bipartite graph with edge weights drawn independently and uniformly from the
interval $[0,1]$. In the following analysis, we consider the BP variant introduced
by Bayati et al~\cite{BipartiteMatching}. Our results can be extended to
other versions of BP for matching and min-cost flow~\cite{SanghaviArbitraryMatching,BayatiArbitraryMatching,BeliefMCF}
in a straightforward way. We discuss these extensions in Section~\ref{sec:OtherBPVars}.

We first discuss the average case, i.e., $\phi=1$, for which we obtain
a lower tail bound of $\Omega(n/t)$ for the probability that more than $t$ iterations
are needed for convergence  (Section~\ref{sec:uniformweights}). For this lower bound,
we use a simple adversarial graph. We leave it as an open problem whether the lower bound
also holds for the complete bipartite graph on $n$ vertices.
After that, we consider complete bipartite graphs with smoothed weights
and prove a lower bound of $\Omega(n\phi/t)$ for the probability that more than $t$ iterations
are needed for convergence (Section~\ref{ssec:smoothedlower}).
We conclude this section with a discussion on how to transfer our results to the other variants
of BP for matching and min-cost flow.

\subsection{Computation Tree}

For proving the lower bounds, we need the notion of a \emph{computation tree},
which we define analogously to Bayati et al.~\cite{BipartiteMatching}.

Let $G=(U \cup V, E)$ be a bipartite graph with $U = \{u_1, \ldots, u_n\}$
and $V = \{v_1,\ldots, v_n\}$.
We denote the level-$k$ \emph{computation tree} with the root labeled $x \in U \cup V$ by $T^k(x)$.
The tree
$T^k(x)$ is a weighted rooted tree of height $k+1$. The root node in $T^0(x)$ has label $x$, its degree is the
degree of $x$ in $G$, and its children are labeled
with the adjacent nodes of $x$ in $G$. $T^{k+1}(x)$ is obtained recursively from $T^k(x)$ by attaching children to every leaf node in
$T^k(x)$. Each child of a former leaf node labeled $y$ is assigned one vertex adjacent
to $y$ in $G$ as a label, but the label of the former leaf node's parent
is not used. (Thus, the number of children is the degree of $y$ minus $1$.)
Edges between nodes with label $u_i$ and label $v_j$ in the computation tree have a weight of $w_{ij}$.

We call a collection $\Lambda$ of edges in the computation tree $T^k(x)$ a \emph{$T$-matching} if no two edges of $\Lambda$
are adjacent in $T^k(x)$ and each non-leaf node of $T^k(x)$ is the endpoint of exactly one edge from $\Lambda$. Leaves
can be the endpoint of either one or zero edges from $\Lambda$. Let $t^k(u_i;r)$ be the weight of a maximum weight $T$-matching in
$T^k(u_i)$
that uses the edge $(u_i,v_r)$ at the root.

\subsection{Average-Case Analysis}
\label{sec:uniformweights}
Consider the undirected weighted complete bipartite graph $K_{2,2}=(U \cup V,E)$, where $U=\{u_1,u_2\}$, $V=\{v_1,v_2\}$, and $(u_i,v_j) \in E$ for $1 \leq i,j \leq 2$. Each edge $(u_i,v_j)=e_{ij}$ has weight~$w_{ij}$ drawn independently and uniformly from $[0,1]$. We define the event $E_{\eps}$ for $0<\eps \leq \frac{1}{8}$ as the event that $w_{11} \in \bigl[\frac{7}{8},1\bigr]$, $w_{12} \in \bigl(\frac{1}{2},\frac{5}{8}\bigr]$, $w_{21} \in \bigl(\frac{5}{8},\frac{3}{4}\bigr]$, and $w_{22} \in [w_{12}+w_{21}-w_{11}-\eps,w_{12}+w_{21}-w_{11})$. Consider the two possible matchings
$M_1=\{e_{11},e_{22}\}$ and $M_2=\{e_{12},e_{21}\}$. If event $E_{\eps}$ occurs, then the weight of $M_2$ is greater than the weight of $M_1$ and the weight difference is at most $\eps$. In addition, $w_{11}$ is greater than $w_{12}$ and the weight difference is at least $1/4$. See Figure~\ref{fig:BipMatching} for a graphical illustration.

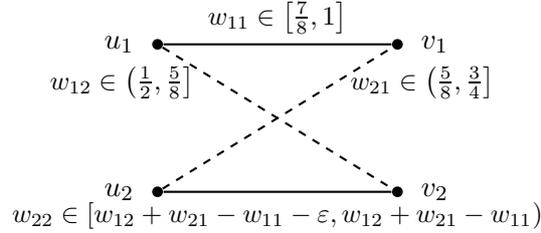
\begin{figure}[th]
\begin{center}
\begin{tikzpicture}[
vertex/.style={fill=black,circle,inner sep =0.5mm},
every label/.style={black,label distance=1mm},
myedge/.style={thick}]
%define nodes
\node [vertex] (u1) [label = {left:$u_1$}] {};
\node [vertex] (u2) [below=1.8cm of u1,label = {left:$u_2$}] {};
\node [vertex] (v1) [right=3cm of u1,label = {right:$v_1$}] {};
\node [vertex] (v2) [below=1.8cm of v1,label = {right:$v_2$}] {};
%define edges
\draw
	 (u1) edge[myedge] node[above,font=\small] {$w_{11} \in \bigl[\frac{7}{8},1\bigr]$} (v1)
	 (u1) edge[myedge,dashed] node[left, near start,font=\small,xshift=-2mm] {$w_{12} \in \bigl(\frac{1}{2},\frac{5}{8}\bigr]$} (v2)
	 (u2) edge[myedge,dashed] node[right, near end,font=\small,xshift=0.6mm] {$w_{21} \in \bigl(\frac{5}{8},\frac{3}{4}\bigr]$} (v1)
	 (u2) edge[myedge] node[below,font=\small] {$w_{22} \in [w_{12}+w_{21}-w_{11}-\eps,w_{12}+w_{21}-w_{11})$} (v2);
\end{tikzpicture}
\caption{If event $E_{\eps}$ occurs, then the weight of the dashed matching $M_2=\{e_{12},e_{21}\}$ is greater than the weight of the solid matching $M_1=\{e_{11},e_{22}\}$ and the weight difference is at most $\eps$. In addition $w_{11}$ is greater than $w_{12}$ and the weight difference is at least $\frac{1}{4}$.}
\label{fig:BipMatching}
\end{center}
\end{figure}

\begin{lemma}
The probability of event $E_{\eps}$ is $\eps/{8^3}$. \label{ProbOfE}
\end{lemma}

\begin{proof}
The intervals in which $w_{11}$, $w_{12}$, and $w_{21}$ have to assume values in order for
event $E_{\eps}$ to occur all have a length of $1/8$. The interval in which $w_{22}$ has to take a value in order for event $E_{\eps}$ to occur, has a length of $\eps$. It is contained completely in the interval $\bigl(0,\frac{1}{2}\bigr]$, since
\[ w_{12}+w_{21}-w_{11}-\eps > \frac{1}{2}+\frac{5}{8}-1-\frac{1}{8}=0 \]
and
\[ w_{12}+w_{21}-w_{11} \leq \frac{5}{8}+\frac{3}{4}-\frac{7}{8}=\frac{1}{2}. \]
Now the probability that $w_{11}$, $w_{12}$, $w_{21}$, and $w_{22}$ all take values in the interval necessary for event $E_{\eps}$ to occur is $\eps/{8^3}$.
\end{proof}

\begin{figure}[th]
\begin{center}
\begin{tikzpicture}[
all nodes/.style={minimum size=0cm,inner sep =0.1mm},
vertex/.style={fill=black,circle,inner sep =0.5mm},
every label/.style={black,label distance=1mm},
myedge/.style={thick}]
	 %left branch
	 \node [vertex] (u1) [label = {above:$u_1$}] {};
	 \node [vertex] (v1) [below left= of u1,label = {left:$v_1$}] {};
	 \node [vertex] (u2) [below left= of v1,label = {left:$u_2$}] {};
	 \node [vertex] (v2) [below left= of u2,label = {left:$v_2$}] {};
	 \node [vertex] (u1p) [below left= of v2,label = {left:$u_1$}] {};
	 \node [vertex] (v1p) [below left= of u1p,label = {left:$v_1$}] {};
	 %right branch
	 \node [vertex] (v2p) [below right= of u1,label = {right:$v_2$}] {};
	 \node [vertex] (u2p) [below right= of v2p,label = {right:$u_2$}] {};
	 \node [vertex] (v1pp) [below right= of u2p,label = {right:$v_1$}] {};
	 \node [vertex] (u1pp) [below right= of v1pp,label = {right:$u_1$}] {};
	 \node [vertex] (v2pp) [below right= of u1pp,label = {right:$v_2$}] {};
	 %edges left branch
	 \draw
	 (u1) edge[myedge] node[left] {$w_{11}$} (v1)
	 (v1) edge[myedge] node[left] {$w_{21}$}(u2)
	 (u2) edge[myedge] node[left] {$w_{22}$} (v2)
	 (v2) edge[dotted,myedge] (u1p)
	 (u1p) edge[myedge] node[left] {$w_{11}$} (v1p)
	 %edges right branch
	 (u1) edge[myedge] node[right] {$w_{12}$}(v2p)
	 (v2p) edge[myedge] node[right] {$w_{22}$} (u2p)
	 (u2p) edge[myedge] node[right] {$w_{21}$} (v1pp)
	 (v1pp) edge[dotted,myedge] (u1pp)
	 (u1pp) edge[myedge] node[right] {$w_{12}$} (v2pp);
	 
\end{tikzpicture}
\caption{The computation tree $T^{4k}(u_1)$.} 
\label{fig:CompTree}
\end{center}
\end{figure}

\begin{lemma}
If event $E_{\eps}$ occurs, then the belief of node $u_1$ of $K_{2,2}$ at the end of the $4k$-th iteration is incorrect for all integers $k \leq \frac{1}{8\eps}-1$. \label{FalseBelief}
\end{lemma}

\begin{proof}
Consider the computation tree $T^{4k}(u_1)$ (see Figure~\ref{fig:CompTree}). According to Bayati~et~al.~\cite[Lemma 1]{BipartiteMatching}, the belief of node~$u_1$ of $K_{2,2}$ after~$4k$ iterations is given by the two-dimensional vector $b^{4k}_{u_1}=\bigl[2t^{4k}(u_1;1)\ \ 2t^{4k}(u_1;2)\bigr]^t$. This means that, after~$4k$ iterations, the belief of node~$u_1$ that it should be matched to~$v_1$ is equal to twice the weight of the maximum-weight $T$-matching of $T^{4k}(u_1)$ that selects edge $(u_1,v_1)$ at the root. Analogously, after~$4k$
iterations, the belief of node~$u_1$ that it should be matched to~$v_2$ is equal to twice the weight of the maximum-weight $T$-matching of $T^{4k}(u_1)$ that selects edge $(u_1,v_2)$ at the root. The maximum-weight $T$-matching $\hat{\Lambda}$ that matches the root node to its child labeled $v_2$, matches each node labeled~$u_1$ to a node labeled~$v_2$ and each node labeled~$u_2$ to a node labeled~$v_1$, since this is the only possible $T$-matching that matches the root node to its child labeled~$v_2$. Define $\Lambda^\star$ as the $T$-matching that matches each node labeled~$u_1$ to a node labeled~$v_1$ and each node labeled~$u_2$ to a node labeled~$v_2$. We show that $\Lambda^\star$ has larger weight than~$\hat{\Lambda}$, which implies that the belief at node~$u_1$ after~$4k$ iterations is incorrect. We have
\begin{align*}
w(\Lambda^\star)-w(\hat{\Lambda})&= (2k+1)w_{11}+2kw_{22}-(2k+1)w_{12}-2kw_{21} &\\
&= 2k(w_{11}+w_{22}-w_{12}-w_{21}) +w_{11}-w_{12} &\\
&\geq -2k\eps +1/4. &
\end{align*}
Now $-2k\eps +1/4$ is greater than zero if $k\leq\frac{1}{8\eps}-1$. 
\end{proof}

\begin{theorem}
The probability that BP for MWM needs at least $t$ iterations to converge for $K_{2,2}$ with edge weights drawn independently and uniformly from $[0,1]$ is at least $\frac{1}{ct}$ for some constant $c > 0$.
\label{thm:twobytwouniform}
\end{theorem}

\begin{proof}
We denote the number of iterations necessary for convergence of BP for MWM by $\tau$. Using Lemma~\ref{ProbOfE} and Lemma~\ref{FalseBelief}, we have
\begin{align*}
\probab(\tau \geq t) &\geq \probab(\tau \geq 4\lceil t/4 \rceil) \geq \probab\left(E_{\frac{1}{8(\lceil t/4 \rceil+1)}}\right) \\
&= \frac{1}{8^4(\lceil t/4 \rceil+1)}\geq \frac{1}{ct}
\end{align*}
for some constant $c > 0$.
\end{proof}

\begin{corollary}
\label{cor:averagelower}
There exist bipartite graphs on $n \geq 4$ nodes, where~$n$ is a multiple of~$4$, with edge weights drawn independently and uniformly from $[0,1]$, for which the probability that BP for MWM needs at least $t$ iterations to converge is $\Omega\bigl(\frac{n}{t}\bigr)$ for $t \geq n/c'$
for some constant $c'>0$.
\label{ConvComplete}
\end{corollary}

\begin{proof}
%In the rest of the proof we assume that $n$ is a multiple of 4. Otherwise, if $n$ is odd, then a perfect matching is impossible. If $n$ is even but not a multiple of 4, then we assign the two excess nodes degree 1 and connect them only with each other.
The bipartite graph consists of $n/4$ copies of $K_{2,2}$ and there are no edges between nodes in different copies of $K_{2,2}$. If BP does not converge in less than~$t$ iterations for at least one of the~$n/4$ copies of $K_{2,2}$, then BP does not converge in less than~$t$ iterations. This holds since a run of BP on this bipartite graph corresponds to~$n/4$ parallel runs of BP on the~$n/4$ copies of $K_{2,2}$.
Using Theorem~\ref{thm:twobytwouniform}, we have that a constant $c>0$ exists such that
\begin{align*}
\probab(\tau < t) &= \left(1-\probab\bigl(\textrm{BP needs at least $t$ iterations for a particular copy of $K_{2,2}$}\bigr)\right)^{n/4} \\
&\leq \left(1-\frac{1}{ct}\right)^{n/4} %\geq 1-\left(\exp\left(-\frac{1}{ct}\right)\right)^{n/4} \\
\leq \exp\left(-\frac{n}{4ct}\right)
\leq 1 - \frac{n}{8ct},
\end{align*}
where the second inequality follows from $1-x \leq \exp(-x)$ and the last inequality follows from $\exp(-x) \leq 1-\frac{x}{2}$ for $x \in [0, 1]$ and from $\frac{n}{4ct} \leq 1$ which holds if $t \geq \frac{n}{4c}$.
\end{proof}

\subsection{Smoothed Analysis}
\label{ssec:smoothedlower}

In this section we consider complete bipartite graphs $K_{n,n}$ in the smoothed setting.
We denote by $X \sim U[a,b]$ that random variable $X$ is uniformly distributed on interval $[a,b]$.
In the following we assume that $\phi \geq 26$ and $n \geq 2$ and even. Similarly to the average case (Section~\ref{sec:uniformweights}),
we define the event $E^{\phi}_{\eps}$ for $K_{2,2}$ and for $0<\eps \leq 1 / \phi$ as the event that $w_{11} \in \bigl[1-\frac{1}{\phi},1\bigr]$, $w_{12} \in \bigl(\frac{23}{26},\frac{23}{26}+\frac{1}{\phi}\bigr]$, $w_{21} \in \bigl(\frac{23}{26},\frac{23}{26}+\frac{1}{\phi}\bigr]$, and $w_{22} \in [w_{12}+w_{21}-w_{11}-\eps,w_{12}+w_{21}-w_{11})$. Consider the two possible matchings $M_1=\{e_{11},e_{22}\}$ and $M_2=\{e_{12},e_{21}\}$. If event $E^{\phi}_{\eps}$ occurs,
then the weight of $M_2$ is greater than the weight of $M_1$ and the weight difference is at most $\eps$. In addition $w_{11}$ is greater than $w_{12}$ and the weight difference is at least $\frac{3}{26}-\frac{2}{\phi}$.

\begin{lemma}
There exist probability distributions on $[0,1]$ for the weights of the edges, whose densities are bounded by $\phi$, such that the probability of event $E^{\phi}_{\eps}$ is at least $\eps\phi/4$. \label{ProbOfEPhi}
\end{lemma}

\begin{proof}
The intervals in which $w_{11}$, $w_{12}$, and $w_{21}$ have to assume
values in order for event $E^{\phi}_{\eps}$ to occur all have a length of $\frac{1}{\phi}$. We choose the corresponding probability distributions such that they have density $\phi$ on the corresponding interval and density 0 elsewhere. The interval in which $w_{22}$ has to assume a value in order for event $E^{\phi}_{\eps}$ to occur
has a length of $\eps$. It is contained completely in the interval $\bigl[\frac{20}{26}-\frac{1}{\phi},\frac{20}{26}+\frac{3}{\phi}\bigr]$, since
\[ w_{12}+w_{21}-w_{11}-\eps > \frac{23}{26}+\frac{23}{26}-1-\frac{1}{\phi}=\frac{20}{26}-\frac{1}{\phi} \]
and
\[ w_{12}+w_{21}-w_{11} \leq
 \left(\frac{23}{26}+\frac{1}{\phi}\right)+ \left(\frac{23}{26}+\frac{1}{\phi}\right)-\left(1-\frac{1}{\phi}\right)
 =\frac{20}{26}+\frac{3}{\phi}.
\]
We choose the probability distribution for $w_{22}$ such that it has density $\frac{\phi}{4}$ on the interval $\bigl[\frac{20}{26}-\frac{1}{\phi},\frac{20}{26}+\frac{3}{\phi}\bigr]$ and 0 elsewhere. Now the probability that $w_{11}$, $w_{12}$, and $w_{21}$ take values in the interval necessary for event $E^{\phi}_{\eps}$ to occur is 1. For $w_{22}$, this probability is $\eps\phi/4$.
This completes the proof of Lemma~\ref{ProbOfEPhi}.
\end{proof}

\begin{lemma}
If event $E^{\phi}_{\eps}$ occurs, then the belief of node $u_1$ at the end of the $4k$-th iteration is incorrect for all integers $k \leq \frac{1}{52\eps}-1$. \label{FalseBeliefSmoothed}
\end{lemma}

\begin{proof}
As in Lemma~\ref{FalseBelief}, a maximum-weight $T$-matching that selects the edge labeled $(u_1,v_1)$ at the root has greater weight than a maximum-weight $T$-matching that selects the edge labeled $(u_1,v_2)$ at the root for these values of~$k$.
\end{proof}

Analogously to the proof of Theorem~\ref{thm:twobytwouniform},
Lemmas~\ref{ProbOfEPhi} and~\ref{FalseBeliefSmoothed}
above immediately yield a lower bound of $\Omega(\phi/t)$ for the probability that
BP runs for at least $t$ iterations.

Our goal in the remainder of this section is to prove an $\Omega(n\phi/t)$ lower bound
for the complete bipartite graph.
Thus, let us consider the complete bipartite graph $K_{n,n} =(U \cup V,E)$ with $U=\{u_p^j \mid p \in \{1,2\},
j \in \{1, \ldots, n/2\}\}$
and $V=\{v_q^j \mid q \in \{1,2\}, j \in \{1,\ldots, n/2\}$.
Let $H^j$ denote the subgraph induced by $\{u_1^j, u_2^j, v_1^j, v_2^j\}$
for $j \in \{1, \ldots, n/2\}$.
The role of the subgraphs $H^j$ is the same as the role of the copies of $K_{2,2}$
in the proof of Corollary~\ref{cor:averagelower}.
Let $e_{pq}^j$ be the edge connecting $u_p^j$ and $v_q^j$ ($p,q\in \{1,2\}$, $j \in \{1, \ldots, n/2\}$).
The weight of this edge is $w^j_{pq}$. We draw edge weights according to the probability distributions
\begin{equation}
\begin{array}{ll@{\qquad}ll}
\displaystyle w^j_{11} & \displaystyle \sim U\left[1-\frac{1}{\phi},1\right],
& \displaystyle w^j_{12} &\displaystyle \sim U\left(\frac{23}{26}, \frac{23}{26}+\frac{1}{\phi}\right], \\[\bigskipamount]
\displaystyle w^j_{21} &\displaystyle \sim U\left(\frac{23}{26}, \frac{23}{26}+\frac{1}{\phi}\right],
& \displaystyle w^j_{22} &\displaystyle \sim U\left[\frac{20}{26}-\frac{1}{\phi},\frac{20}{26}+\frac{3}{\phi}\right], \\[\bigskipamount]
\displaystyle w_{ab} &
\multicolumn{3}{l}{\displaystyle \sim U\left[0,\frac{1}{\phi}\right] \text{ if $u_a\in H^j$ and $v_b\in H^k$ with $j \neq k$.}}
\end{array}
\label{smoothweights}
\end{equation}
We call the edges between nodes in the same induced subgraph $H^j$ \emph{heavy edges}. Edges between nodes in different subgraphs $H^j$ and $H^k$ we call \emph{light edges}.
By assumption, we have $\phi \geq 26$. Thus, the weight of any light edge is at most $1/26$, while every
heavy edge weighs at least $19/26$.

In contrast to the proof of Corollary~\ref{cor:averagelower}, we now have to make sure that
light edges are not used in any computation tree.
This allows us to prove the lower bound in a similar way as Theorem~\ref{thm:twobytwouniform}
and Corollary~\ref{cor:averagelower}.

\begin{lemma}
Let $\Lambda^\star$ be the maximum-weight $T$-matching on the computation tree $T^k(u_i)$.
Then $\Lambda^\star$ does not contain any light edges.
\label{NoLightEdges} 
\end{lemma}

\begin{proof}
Assume to the contrary that $\Lambda^\star$ contains a light edge $(x,y)$. In that case $x$ and~$y$ are in different subgraphs. The idea of the proof is to construct a path~$P$ from one leaf of the computation tree to another leaf that includes edge $(x, y)$. Path~$P$ alternately consists of edges that are in $\Lambda^\star$ and edges that are not. We show that a new $T$-matching of greater weight can be constructed by removing from $\Lambda^\star$ the edges in $P \cap \Lambda^\star$ and adding the edges in $P \setminus \Lambda^\star$.

We include the edge labeled $(x,y)$ in~$P$ and extend~$P$ on both sides: We start with node $z_0 = x$ and node $z_0 = y$, respectively, and construct the corresponding part of~$P$ as follows:

\newcommand{\QUADSIZE}{3ex}
\newcommand{\QUAD}{\hspace{\QUADSIZE}}
\begin{enumerate}
  \item \textbf{for} $i = 1, 3, 5, \ldots$ \textbf{do}
  \item \QUAD \textbf{if} $z_{i-1}$ is a leaf node \textbf{then} terminate.
  \item \QUAD Let~$H^k$ be the subgraph that~$z_{i-1}$ belongs to.
  \item \QUAD \begin{minipage}[t]{\textwidth-\leftmargin-\QUADSIZE}%
        Let $e_i = (z_{i-1}, z_i)$ be the edge incident to~$z_{i-1}$ that belongs to the optimal matching with respect to~$H^k$.
        \end{minipage}\label{line:stepI}
  \item \QUAD Add~$e_i$ to~$P$.
  \item \QUAD \textbf{if} $z_i$ is a leaf node \textbf{then} terminate.
  \item \QUAD \begin{minipage}[t]{\textwidth-\leftmargin-\QUADSIZE}%
        Let $e_{i+1} = (z_i, z_{i+1})$ be the (unique) edge incident to~$z_i$ that belongs to~$\Lambda^\star$.
        \end{minipage}\label{line:stepII}
  \item \QUAD Add $e_{i+1}$ to~$P$.
\end{enumerate}
It is clear that the procedure can only terminate if it finds a leaf. Moreover, the constructed sequence is alternating. Now we can show that no node will be visited twice: Otherwise there was an index~$i$ such that $z_{i-1} = z_{i+1}$ since we are moving in a tree. This can not happen, however, as the sequence is alternating. Hence, the procedure terminates. With the previous properties we also obtain that both paths constructed starting with $z_0 = x$ and $z_0 = y$, respectively, are disjoint since $z_1 \notin \{x, y\}$ in both cases. Consequently, we obtain one simple path~$P$ connecting two distinct leaf nodes and containing edge $(x, y)$.

We now show that the weight of the edges in $P \setminus \Lambda^\star$ is strictly larger than the weight of the edges in $P \cap \Lambda^\star$. For this, let~$P$ be of the form $P = (p_0, \ldots, p_\ell)$, where~$\ell$ is even and where $(p_0, p_1) \in \Lambda^\star$. Let $I \subseteq \{ 1, \ldots, \ell \}$ be the set of indices~$i$ for which $(p_{i-1}, p_i)$ is a light edge. Clearly, $(p_{i-1}, p_i) \in \Lambda^\star$ for each $i \in I$ by construction (see Line~\ref{line:stepI}). Since the light edge $(x, y)$ belongs to~$P$ we have $I \neq \emptyset$. For $i \in I$ let $P_i = (p_{i-1}, p_i, p_{i+1})$ be the subpath of~$P$ of length~$2$ starting at node~$p_{i-1}$. As $(p_i, p_{i+1})$ is a heavy edge, $w_{p_i p_{i+1}} - w_{p_{i-1} p_i} \geq \big( \frac{20}{26} - \frac{1}{\phi} \big) - \frac{1}{\phi} = \frac{20}{26} - \frac{2}{\phi}$. Hence, the difference in weight between the edge of~$P_i$ that belongs to $\Lambda^\star$ and the other one is significant.

Now remove all paths~$P_i$ from~$P$ and consider the subpaths of~$P$ (connected components) that remain. There are at most $|I|+1$ such subpaths~$P'$, each has even length, and they only consist of heavy edges, i.e., all their edges lie in one subgraph~$H^k$ where~$k$ depends on~$P'$. Consider such a subpath~$P'$ and partition it into subpaths~$\tilde{P}_j$ of length~$4$ and, if the length of~$P'$ is no multiple of~$4$, into one subpath~$\hat{P}$ of length~$2$. The $\Lambda^\star$-edges of~$\tilde{P}_j$ form the non-optimal matching on~$H^k$, whereas the other two edges form the optimal matching on~$H^k$. Hence, the total weight of $\tilde{P}_j \cap \Lambda^\star$ is at most the total weight of $\tilde{P}_j \setminus \Lambda^\star$. Only for~$\hat{P}$ we might have the case that the weight of $\hat{P} \cap \Lambda^\star$ is larger than the weight of $\hat{P} \setminus \Lambda^\star$, but since both edges are heavy, the difference is at most $1 - \big( \frac{20}{26} - \frac{1}{\phi} \big) = \frac{6}{26} + \frac{1}{\phi}$. Hence, the difference between the total weight of $P \setminus \Lambda^\star$ and the total weight of $P \cap \Lambda^\star$ is at least
\[
  |I| \cdot \left( \frac{20}{26} - \frac{2}{\phi} \right) - (|I|+1) \cdot \left( \frac{6}{26} + \frac{1}{\phi} \right)
  = |I| \cdot \left( \frac{14}{26} - \frac{3}{\phi} \right) - \left( \frac{6}{26} + \frac{1}{\phi} \right)
  \geq \frac{4}{26}
  > 0
\]
since $|I| \geq 1$ and $\phi \geq 26$.

We can now construct a $T$-matching with heavier weight than $\Lambda^\star$ by removing the edges in $P \cap \Lambda^\star$ from $\Lambda^\star$ and adding the edges in $P \setminus \Lambda^\star$. This contradicts the assumption that the maximum weight $T$-matching includes a light edge and proves the lemma.
\end{proof}

\begin{theorem}
\label{thm:smoothedlower}
There exist probability distributions on $[0,1]$ for the weights of the edges, whose densities are bounded by $\phi$, such that the probability that BP for MWM needs at least $t$ iterations to converge for $K_{n,n}$ is $\Omega(n\phi/t)$ for $t \geq n\phi/c$ for some constant $c>0$.
\end{theorem}

\begin{proof}
We choose the probability distributions for the edge weights according to~\eqref{smoothweights}. Let
$\eps=\frac{1}{52(k+1)}$ for $k = 4\lceil t/4\rceil$ and assume that event $E^{\phi}_{\eps}$ occurs for subgraph $H^j$. In this case the weight of matching $M_2 = \{ e^j_{12}, e^j_{21} \}$ is larger than matching $M_1 = \{ e^j_{11}, e^j_{22} \}$, but at most by the small amount of~$\eps$. Consider the computation tree $T^{4k}(u^j_1)$. As in the proof of Lemma~\ref{FalseBelief} we know that if the maximum weight $T$-matching $\Lambda^\star$ on $T^{4k}(u^j_1)$ does not include the edge labeled $e^j_{12}$ at the root, then BP has not yet converged within the first $4k \geq t$ iterations (see Bayati~et~al.~\cite[Lemma 1]{BipartiteMatching}).

%Suppose $\Lambda^\star$ contains the edge labeled $e^j_{12}$ at the root. We know from Lemma~\ref{NoLightEdges} that a maximum weight $T$-matching does not contain any light edges.
%As in the proof of Lemma~\ref{NoLightEdges}, we can uniquely define the path~$P$ from a leaf of the computation tree up to the root and down to another leaf of the computation tree that consists of the %edges $e^j_{12}$ and $e^j_{11}$ at the root and then alternately of edges in $\Lambda^\star$ and edges in $\{e^j_{11},e^j_{22}\}$ until the path ends in a leaf of $T^{4k}(u^j_1)$ on both sides.
%Using Lemma~\ref{FalseBeliefSmoothed},
%we have that the weight of the edges in $P \setminus \Lambda^\star$ is greater than the weight of the edges in $P \cap \Lambda^\star$. We can obtain a $T$-matching of greater weight by removing from %$\Lambda^\star$ the edges in $P \cap \Lambda^\star$ and by adding the edges in $P \setminus \Lambda^\star$. This contradicts the assumption that $\Lambda^\star$ contains the edge labeled $e^j_{12}$ at the %root and therefore the belief of node $u^j_1$ after $4k$ iterations is incorrect.
%
%
We show that $\Lambda^\star$ does not include edge $e^j_{12}$. Assume to the contrary that it does. We know from Lemma~\ref{NoLightEdges} that~$\Lambda^\star$ does not contain light edges. Now we use the same procedure to create a path~$P$ from one leaf of $T^{4k}(u^j_1)$ to another leaf that contains edge $e^j_{12}$ and alternates between edges from~$\Lambda^\star$ and edges from $T^{4k}(u^j_1) \setminus \Lambda^\star$. Since $T^{4k}(u^j_1)$ has height $4k+1$ and since $u^j_1$ is the root of $T^{4k}(u^j_1)$, path~$P$ contains exactly $8k+2$ edges, $2k+1$ of which are edges $e^j_{12}$, $2k+1$ of which are edges $e^j_{11}$, $2k$ of which are edges $e^j_{22}$, and $2k$ of which are edges $e^j_{21}$. The edges $e^j_{12}$ and $e^j_{21}$ are exactly the edges of $P \cap \Lambda^\star$. As in Lemma~\ref{FalseBelief}, the difference of weight between edges from $P \setminus \Lambda^\star$ and $P \cap \Lambda^\star$ is at least
\begin{align*}
  w^j_{11} - w^j_{12} - 2k\eps
  &\geq \left( \left( 1 - \frac{1}{\phi} \right) - \left( \frac{23}{26} + \frac{1}{\phi} \right) \right) -\frac{2k}{52(k+1)} \cr
  &> \frac{3}{26} - \frac{2}{\phi} - \frac{1}{26}
  \geq 0
\end{align*}
since $\phi \geq 26$. This contradicts the fact that~$\Lambda^\star$ is optimal since removing from~$\Lambda^\star$ the edges in $P \cap \Lambda^\star$ and adding the edges in $P \setminus \Lambda^\star$ yields a $T$-matching of heavier weight for $T^{4k}(u^j_1)$.

We have shown that BP does not converge within the first~$t$ iterations if event $E^{\phi}_{\eps}$ occurs for some subgraph $H^j$. Since there are $n/2$ such subgraphs, we get that the probability that BP for MWM needs at least $t$ iterations to converge for $K_{n,n}$ is $\Omega\bigl(\frac{n\phi}{t}\bigr)$
since
\begin{align*}
  \probab(\tau \leq t)
  &\leq \probab \big( \text{$E^{\phi}_{\eps}$ does not occur for any subgraph $H^j$} \big) \cr
  &\leq \left( 1 - \frac{\eps \phi}{4} \right)^{n/2}
  \leq \exp \left( -\frac{\eps n \phi}{8} \right)
  = \exp \left( -\frac{n\phi}{8 \cdot 52 \cdot (4 \cdot \lceil t/4 \rceil+1)} \right) \cr
  &\leq 1 - \frac{n\phi}{2 \cdot 8 \cdot 52 \cdot (4 \cdot \lceil t/4 \rceil+1)}
\end{align*}
where the second inequality follows from Lemma~\ref{ProbOfEPhi}. The third inequality is due to the fact that $1-x \leq \exp(-x)$, whereas the last inequality stems from $\exp(-x) \leq 1 - \frac{x}{2}$ for $x \in [0, 1]$. If $x = \frac{n\phi}{8 \cdot 52 \cdot (4 \cdot \lceil t/4 \rceil + 1)}$ is at most~$1$, which holds for $t \geq \frac{n\phi}{8 \cdot 52}$, then the correctness follows.
\end{proof}

Note that the lower bound on the probability that $BP$ for $MWM$ converges within $t$ iterations only differs a factor $O(m)$ from the upper bound from Section~\ref{ssec:uppermatching}.

\subsection{Other Variants of Belief Propagation}
\label{sec:OtherBPVars}
The results of Sections~\ref{sec:uniformweights} and~\ref{ssec:smoothedlower} also hold for other versions of belief propagation for minimum/maximum-weight (perfect) $b$-matching and min-cost flow~\cite{BayatiArbitraryMatching,SanghaviArbitraryMatching,BeliefMCF} applied to the matching problem on bipartite graphs. The difference in the number of iterations until convergence differs no more than a constant factor. We omit the technical details but provide some comments on how the proofs need to be adjusted.   

Some of the versions of BP consider minimum-weight perfect matching~\cite{BayatiArbitraryMatching} or min-cost flow~\cite{BeliefMCF} instead of maximum-weight perfect matching. For these versions we get the same results if we have edge weights $\tilde{w}_{e}=1-w_{e}$ for all edges $e$.

For some of the versions of BP~\cite{SanghaviArbitraryMatching,BeliefMCF} the root of the computation tree is an edge instead of a node . If we choose the root of this tree suitably, then we have that the difference in weight between the two matchings $M_1$ and $M_2$ of at most $\eps$ not only has to `compensate' the weight difference $\Delta w(e_1,e_2)$ between an edge $e_1$ in $M_1$ and an edge $e_2$ in $M_2$, but the entire weight $w_e$ of an edge $e$ in $M_1$ or $M_2$. However, the probability distributions for the edge weights in Sections~\ref{sec:uniformweights} and~\ref{ssec:smoothedlower} are chosen such that $\Delta w(e_1,e_2)$ and $w_e$ do not differ more than a constant factor.

%\bibliographystyle{plain}
%\bibliography{biblio}

\end{document}